\documentclass[11pt]{amsart}
\usepackage{amsmath,mathrsfs,amssymb,appendix}
\usepackage{slashed}

\usepackage{tikz}
\usepackage{caption,subcaption}
\usepackage{amsaddr}
\usetikzlibrary{snakes}
\theoremstyle{plain}
\newtheorem{theorem}{Theorem}[section]
\newtheorem{lemma}[theorem]{Lemma}

\newtheorem{corollary}[theorem]{Corollary}

\newtheorem{remark}[theorem]{Remark}

\theoremstyle{definition}

\numberwithin{equation}{section}
\usepackage[top=1in, bottom=0.9in, left=1.15in, right=1.15in]{geometry}

\def\al{\begin{align}}
\def\eal{\end{align}}
\def\be{\begin{equation}}
\def\ee{\end{equation}}
\begin{document}

\title[The Null-Timelike Problems of Linear Wave Equations]
{The Null-Timelike Boundary Problems of Linear Wave Equations in  Asymptotically Anti-de Sitter Space}
\author[X. Wu]{Xiao-ning Wu}
\address{Institute of Mathematics, Chinese Academy of Sciences\\
Beijing, 100190, China}
\email{wuxn@amss.ac.cn}
\author[L. Zhang]{Lin Zhang}
\address{Institute of Mathematics, Chinese Academy of Sciences\\
Beijing, 100190, China}
\email{linzhang@amss.ac.cn}

\begin{abstract}
In this paper, we study the linear wave equations in an asymptotically anti-de Sitter spacetime. We will consider the mixed boundary problem, where the initial data are given on an outgoing null hypersurface and a timelike hypersurface, and the asymptotic information is given on conformal infinity.
\end{abstract}

\maketitle

\section{Introduction}
Initial value problem is one of central problems of mathematical physics equations. Since the requirements of different physical problems, many different kinds of initial problems have been studied. For example, characteristic initial value problems\cite{CIV}, initial-boundary value problems\cite{IBVP}. In this paper, null-timelike initial-boundary value problem will be studied. This problem was first noticed by R. Bartnik\cite{Bartnik97}. He tried to solve Einstein equations by imposing boundary data on a timleke cylinder and a outgoing light cone which started from a space-like two sphare. One years later, Balean solved a simple model of such problem, i.e. null-timelike boundary problem for linear wave equations in Minkowski space\cite{Balean1997}. In following paper, Balean and Bartnik generalize Balean's result to Maxwell field in Minkowski space. Since their method strongly depends on Minkowski back ground, it is nontrivial to generalize their work to general back ground. This is one motivation of this work. Another motivation of this paper is try to understand AdS/CFT correspondence\cite{Ma98} in terms of initial-boundary value problem. It is well-known that AdS/CFT correspondence is a great breakthrough in theoretical physics. At the very beginning of AdS/CFT correspondence, Witten\cite{Wit98} understood it in terms of boundary problem of elliptic partial differential equations(PDE) since he considered the Euclidean version of this conjecture. In a recent work, Witten reconsider this problem and give some more careful conjecture on how to choose a suitable boundary condition\cite{Wit18}. In past decades, more and more interest has been focus on how to use this conjecture solving practical physical problems. One of the most active idea is to study condensed matter physics by using AdS/CFT correspondence, for example, AdS/QCD theory\cite{ADSQCD}, AdS/condense matter theory\cite{3H08,ADSCMT}, fluid/gravity correspondence\cite{Strominger}, $\cdots$. The main idea of these work is considering dynamical process on asymptotic AdS black hole back ground and analog associated physical properties by using AdS/CFT correspondence. Since the perturbation is imposed from time-like boundary or conformal boundary of asymptotic AdS spacetime\cite{Strominger}, the null-timelike boundary value problem will give a mathematical foundation for those work.

This paper is organized as following : in next section, some back ground knowledge and statement of main result is given. In section III, some necessary notations and conformal transformations are introduced, which will be used in the proof. Section IV is the key step of this paper. Energy-estimates is given in this section. These results are key tools for our proof. With theorems of section IV, we finish the proof of our main theorem in section V. Section VI contains some discussions on our results.

\section{Back ground and statement of main theorem}
Let $(M,g)$ is a Lorentzian manifold, where $g$ is an asymptotically anti-de Sitter metric. In the Bondi-Sachs coordinates $\{\tau, r, x_2,x_3\}$, where $\tau=constant$ are outgoing null hypersurface,  $\{x_2,x_3\}$ forming local coordinates on $S_{\tau,r}$( i.e. the 2-dimension surfaces with $\tau=constant$ and $r=constant$), and $4\pi r^2=Area(S_{\tau,r})$. We suppose that  $g$ is given by
\be
g=-Ve^{2\eta}\text{d}\tau^2
-2e^{2\eta}\text{d}\tau\text{d}r+r^2h_{AB}(\text{d}x^A-U^A\text{d}\tau)
(\text{d}x^B-U^B\text{d}\tau),
\ee
where, $V$, $\eta$, $U^A$, $h_{AB}$ are functions of $\tau, r, x^A$ with $A=2,3$, and $\det (h_{AB})=1$. And, furthermore, as $r\to +\infty$,
\be\label{asym-V}
V=1-\frac{\Lambda}{3}e^{\zeta}r^2+O(\frac{1}{r}),\ \ \Lambda<0,
\ee
and
\be\label{asym-U}
\lim_{r\to \infty}rU^A=\lim_{r\to \infty}\zeta=\lim_{r\to \infty}\eta=0,
\ee
and
\be
h_{AB}\text{d}x^A\text{d}x^B\to g_{S^2},
\ee
uniformly on $[0,T]\times S^2$, where $g_{S^2}$ is the standard round metric on $S^2$.

In this paper, we study the linear wave equations
\be\label{wave-1}
\Box_g u=0,
\ee
with boundary condition
\begin{align}\label{wave-2}\begin{split}
&u|_{\Sigma_0=\{\tau=0\}}=\frac{\varphi e^{-\eta}}{r},\\
&u|_{\mathcal{T}=\{r=R\}}=\frac{\psi_1 e^{-\eta}}{R},
\end{split}\end{align}
for some given functions $\varphi$ on $\Sigma_0$ and $\psi_1$ on $\mathcal{T}$. Furthermore, we need asymptotic condition on $u$ at further null infinity $\mathcal{I}$ (with topology $[0,+\infty)\times S^2$),
\be\label{wave-3}
\lim_{r\to \infty}e^{\eta}ru=\psi_2,
\ee
for the given functions $\psi_2$ on $[0,+\infty)\times S^2$.
\begin{figure}
\begin{center}
\begin{tikzpicture}
\draw[thin]
(0,3.5)--(0,0);
\node at (-0.25,1.3) {$\mathcal{T}$};
\draw[thin]
(0,0)--(4,4);
\node at (2.4,2) {$\Sigma_0$};
\draw[dashed]
(4,4)--(4,6.5);
\node at (4.2,5.6) {$\mathcal{I}$};
\draw[thin]
(0,2)--(4,6);
\node at (0.95,3.3) {$\Sigma_T$};
\draw[thin]
(4,4) sin (2.1,5.5);
\node at (2.8,5.5) {$\mathcal{H}$};
\node at (2,2.9) {$\mathcal{M}_1$};
\node at (3.6,5.0) {$\mathcal{M}_2$};
\end{tikzpicture}
\end{center}
\caption{}
\label{fig:1}
\end{figure}

We will prove the following result,
\begin{theorem}
Suppose $g$ sufficiently regular and with conformal regular condition (see later). Assume $\psi_1\in H^{2k}([0,T]\times S^2)$, $\psi_2\in H^k([0,T]\times S^2)$, and
$$\|\varphi\|_{\widetilde H^{2k}((R, \infty)\times S^2)}=
\sum_{|\alpha|\leq 2k,\alpha=(\alpha_1,\alpha_2,\alpha_3)}\big(\int_R^{\infty}\int_{{S}^{2}}
\frac{|\partial^{\alpha}\varphi|^{2}}{r^{2+4\alpha_1}}drd\Sigma\big)^{\frac{1}{2}}<+\infty.$$
Then, there exists a unique solution $u$, which satisfies \eqref{wave-1}-\eqref{wave-2}-\eqref{wave-3} in $[0,T]\times[R,+\infty)\times S^2$,
and
\begin{align}\begin{split}
&\|e^{\eta}ru\|_{\widetilde H^{k}([0,T]\times(R, \infty)\times S^2)}=
\sum_{|\alpha|\leq k,\alpha=(\alpha_0,\alpha_1,\alpha_2,\alpha_3)}\big(\int_0^T\int_R^{\infty}\int_{{S}^{2}}
\frac{|\partial^{\alpha}(e^\eta ru)|^{2}}{r^{2+4\alpha_1}}d\Sigma dr d\tau\big)^{\frac{1}{2}}\\
\leq& C\{\|\varphi\|_{\widetilde H^{2k}((R, \infty)\times S^2)}+\|\psi_1\|_{H^{2k}([0,T]\times S^2)}+\|\psi_2\|_{H^{k}([0,T]\times S^2)}\},
\end{split}\end{align}
where $C$ is a constant depending on $g$ and $k$.
\end{theorem}
\section{Conformal transformation}
 In the new coordinates
$$(\tau, z=\frac{1}{r}, x_2, x_3),$$
we have
\be
g=-Ve^{2\eta}d\tau^{2}+\frac{2}{z^{2}}e^{2\eta}d\tau dz
+\frac{1}{z^{2}}h_{AB}(dx^{A}-U^{A}d\tau)(dx^{B}-U^{B}d\tau).
\ee
We introduce the conformal metric $\hat{g}$ given by
\begin{equation}\label{eq-metric-g-bar}
\hat{g}=\Xi^{2}{g}=-z^2Vd\tau^{2}+2d\tau dz
+h_{AB}e^{-2\eta}(dx^{A}-U^{A}d\tau)(dx^{B}-U^{B}d\tau),
\end{equation}
where $\Xi=ze^{-\eta}$.
\begin{remark} We say $g$ satisfies the conformal $C^k$ condition, which means the conformal metric $\hat{g}$ can be $C^k$ extended to $z=0$.
\end{remark}
\begin{remark}
By \eqref{asym-V}, we have, $z^2V=-\frac{\Lambda}{3}$ and $\hat{g}(\nabla z,\nabla z)=-\frac{\Lambda}{3}>0$ at $z=0$, hence, the hypersurface $\{z=0\}$ is timelike in $(M,\hat{g})$.
\end{remark}

Set $v=\Xi^{-1}u$, then
\begin{align}\begin{split}
\Box_{g}u&=\frac{1}{(-\det {g})^{{1}/{2}}}\sum\limits_{i=0}^{3}\partial_{i}[(-\det {g})^{{1}/{2}}{g}^{ij}\partial_ju]\\
&=\Xi^3\Box_{\hat{g}}v+[\Xi^2\Box_{\hat{g}}\Xi
-2\Xi\hat{g}^{ij}\partial_i\Xi\partial_j\Xi]v.
\end{split}\end{align}
By $\partial_i\Xi=\delta_{i1}e^{-\eta}-\Xi\partial_i\eta$, then
$$\hat{g}^{ij}\partial_i\Xi\partial_j\Xi=\Xi^2V-2\Xi e^{-\eta}\hat{g}^{1j}\partial_j\eta+\Xi^2\hat{g}^{ij}\partial_i\eta\partial_j\eta,$$
and
$$\Box_{\hat{g}}\Xi=\frac{1}{\kappa}\partial_i\kappa \hat{g}^{1i}e^{-\eta}
+\partial_i\hat{g}^{1i}e^{-\eta}-2e^{-\eta}\hat{g}^{1i}\partial_i\eta
-\Xi(\Box_{\hat{g}}\eta+\hat{g}^{ij}\partial_i\eta\partial_j\eta),$$
where $\kappa=e^{-2\eta}\sqrt{\gamma}, \gamma=\det(h_{AB}).$  Hence,
\begin{align}\begin{split}
&\Xi^2\Box_{\hat{g}}\Xi
-2\Xi\hat{g}^{ij}\partial_i\Xi\partial_j\Xi=\Xi^3\omega\\
=&\Xi^3[\frac{1}{2z}\partial_A(\log\gamma) U^A+\frac{1}{z}\partial_A(U^A)-\Box_{\hat{g}}\eta-3\hat{g}^{ij}\partial_i\eta\partial_j\eta],
\end{split}\end{align}
where, we have used that $\gamma$ is independent on $\tau$ and $r$. By \eqref{asym-U}, in particular $\lim\limits_{r\to 0}rU^A=0$, then wave equations reduces to the following equation for $v$,
$$\Box_{\hat{g}}v+\omega v=0.$$

From now on, we work on the Lorentzian manifold $(\bar{M}, \hat{g})$ with boundary, which is un-physical spacetime, and  has the topology $[0,+\infty)\times[0,\frac{1}{R}]\times S^2$ and a global coordinate system
$\{\tau,z,x_2,x_3\}$ with $\tau\in [0,+\infty)$, $z\in[0,\frac{1}{R}]$ and $\{x_2,x_3\}$ forming local coordinates on $S^2$. For some fixed $T>0$, we set
$$\Omega_{T}=\{ (\tau, z)|0<\tau< T, 0<z< \frac{1}{R}\}\times S^{2},$$
and
\begin{align*}
\Sigma_{0}&=\{(0,z)|0\leq z\leq \frac{1}{R}\}\times S^{2},\\
\Sigma_{T}&=\{(T,z)|0\leq z\leq \frac{1}{R}\}\times S^{2},\\
\mathcal{T}&=\{(\tau, \frac{1}{R})|0\leq\tau\leq T\}\times S^{2},\\
\mathcal{I}&=\{(\tau,0)|0\leq\tau< +\infty\}\times S^2,
\end{align*}
where, $\mathcal{I}$ represent the further null infinity of the physical spacetime $(M,g)$, and is timelike.

We introduce two null vector fields
$$N_1=\nabla\tau=\partial_{z},\quad
N_{2}=-(\nabla z-\frac{1}{2}z^2V\nabla \tau)=-(\partial_{\tau}+\frac{1}{2}z^{2}V\partial_{z}+g^{1A}\partial_{A}),$$
and, it is easy to check
$$\hat{g}(N_1,N_1)=\hat{g}(N_2,N_2)=0, \ \ \hat{g}(N_1,N_2)=-1.$$
\begin{remark}
$N_1,N_2$ are further point null vector fields.
\end{remark}
We introduce the ingoing null cone $H_\mu$, which is generated by null geodesics $\gamma(t)$ starting  from $p$ on  2-sphere $S_{0,\mu}$, with $\gamma(0)=p$ and $\dot{\gamma}(0)=N_2(p)$. Furthermore, we definite following domain
$$H_{(\nu,\frac{1}{R})}=\bigcup_{\mu\in(\nu,\frac{1}{R})}H_{\mu},$$
and
$$\Omega_{T,\nu}=\Omega_{T}\cap H_{(\nu,\frac{1}{R})}.$$
\begin{remark}
For fixed $T$, we take $\nu$ small enough, then, the boundary of $\Omega_{T,\nu}$ is composed by four regular hypersufaces,  and which located on $\Sigma_0$, $\Sigma_T$, $\mathcal{T}$, and $H_{\nu}$.
\end{remark}
\begin{remark}
For convenience, we denote $H_{0}$ by $\mathcal{H}$, and denote $\mathcal{M}_1=H_{(0,\frac{1}{R})}$. In fact, $\mathcal{M}_1$ is the maximal determined space in physical spacetime by $\Sigma_0$ and $\mathcal{T}$.
\end{remark}
\begin{remark}
We denote $\mathcal{M}_2$ for the domain bounded by $\mathcal{H}$ and $\mathcal{I}$.
\end{remark}

\section{Energy estimate}
In this section, we work on the un-physical spacetime $(M,\hat{g})$, for convenience, we also denote $\hat{g}$ by $g$. We recall the wave equations
\be\label{wave-equ}
\Box_{g} v+\omega v=0,
\ee
and, we consider the boundary condition
\begin{align}\label{wave-boundary}\begin{split}
v|_{\Sigma_0}&=\varphi,\\
v|_{\mathcal{T}}&=\psi_1,\\
v|_{\mathcal{I}}&=\psi_2.
\end{split}\end{align}
\begin{remark}
We can give a rough view for the solvable of this mixed boundary problem. First, by $\varphi$ and $\psi_1$, we can solve the wave equations in $\mathcal{M}_1$, particularly, we can get the information of $v$ on $\mathcal{H}$, and then, combine $\psi_2$, we can solve v in $\mathcal{M}_2$.
\end{remark}

For a $C^1$-function $\phi$, the {\it associated energy momentum tensor} $Q[\phi]$ is a symmetric 2-tensor
defined by, for any vector fields $X$ and $Y$,
\begin{equation*}
Q[\phi](X,Y)=(X\phi)(Y\phi)-\frac{1}{2}g(X,Y)|\nabla\phi|^{2}-g(X,Y)\phi^{2}.
\end{equation*}

\begin{lemma}\label{lem}
Let $\phi$ be a given $C^{2}$-function and $Q[\phi]$ be the associated energy momentum tensor.
Set $X=a_{1}N_{1}+a_{2}N_{2}$ and $Y=a_{3}N_{1}+a_{4}N_{2}$, for some $a_i$, $i=1,2,3,4$. Then,
\begin{align*}
Q[\phi](X, Y)= a_{1}a_{3}(\partial_{z}\phi)^{2}+a_{2}a_{4}(N_{2}\phi)^{2}
+(a_{1}a_{4}+a_{2}a_{3})(\frac{1}{2}g^{AB}\partial_{A}\phi\partial_{B}\phi+\phi^{2}).
\end{align*}
Moreover, if $a_{i}\geq0$, $i=1,2,3,4$, then
$Q[\phi](X,Y)\geq0$.
\end{lemma}

\begin{proof}
First, we have
\begin{align*}
Q[\phi](X, Y)= a_{1}a_{3}Q[\phi](N_{1}, N_{1})+a_{2}a_{4}Q[\phi](N_{2}, N_{2})+(a_{1}a_{4}+a_{2}a_{3})Q[\phi](N_{1}, N_{2}).
\end{align*}
Next,
\begin{equation*}
|\nabla\phi|^{2}=2\partial_{\tau}\phi\partial_{z}\phi+z^{2}V(\partial_{z}\phi)^{2}+
2g^{1A}\partial_{z}\phi\partial_{A}\phi+g^{AB}\partial_{A}\phi\partial_{B}\phi,
\end{equation*}
and then
\begin{align*}
&Q[\phi](N_{1}, N_{1})=(\partial_{z}\phi)^{2}, \
Q[\phi](N_{2}, N_{2})=(N_{2}\phi)^{2}, \\
&Q[\phi](N_{1}, N_{2})=\frac{1}{2}g^{AB}\partial_{A}\phi\partial_{B}\phi+\phi^{2}.
\end{align*}
A simple substitution yields the desired results.
\end{proof}

We take a tetrad vector fields
$$\{N_1,N_2,e_3,e_4\},$$
where, $e_3,e_4$ are vector fields tangent to $S_{\tau,z}$, and satisfy,
$$g(e_3,e_3)=g(e_4,e_4)=1,\ g(N_1,e_4)=g(N_1,e_3)=g(N_2,e_4)=g(N_2,e_3)=g(e_4,e_3)=0.$$
Take a null vector field
$$L=l_1N_1+l_2N_2+l_3e_3+l_4e_4,$$
by $g(L,L)=0$, we have
\be
-2l_1l_2+l_3^2+l_4^2=0.
\ee
\begin{lemma}
Let $L=l_1L_1+l_2N_2+l_3N_3+l_4e_4$ be a null vector with $l_1\geq0$, $l_2>0$, and $Y=b_1N_1+b_2N_2$ with $b_1\geq0, b_2\geq0$, then, we have $Q[\phi](L,Y)\geq0$.
\end{lemma}
\begin{proof}
First, by $g(L,L)=0$, we have
 $$2l_1l_2=l_3^2+l_4^2.$$
And
\begin{align}\begin{split}
&Q[\phi](L,Y)\\
=&b_1l_1Q[\phi](N_1,N_1)+b_2l_2Q[\phi](N_2,N_2)+(b_2l_1+b_1l_2)Q[\phi](N_1,N_2)\\
&+l_3b_1(e_3\phi)(N_1\phi)+l_3b_2(e_3\phi)(N_2\phi)
+l_4b_1(e_4\phi)(N_1\phi)+l_4b_2(e_4\phi)(N_2\phi)\\
=&\frac{1}{2}b_1l_2[(\frac{l_3}{l_2}N_1\phi+e_3\phi)^2+(\frac{l_4}{l_2}N_1\phi+e_4\phi)^2]
+b_2l_2(N_2\phi+\frac{l_3}{2l_2}e_3\phi+\frac{l_4}{2l_2}e_4\phi)^2\\
&+\frac{1}{4}b_2l_2(\frac{l_3}{l_2}e_4\phi-\frac{l_4}{l_2}e_3\phi)^2+(b_2l_1+b_1l_2)\phi^2,
\end{split}\end{align}
hence, we know $Q[\phi](L,Y)\geq0$.
\end{proof}
\begin{remark}
A straightforward calculation yields
\begin{align}\begin{split}
&g(L,\frac{l_3}{l_2}N_1+e_3)=0,\ \ g(L,\frac{l_4}{l_2}N_1+e_4)=0,\\
&g(L,N_2+\frac{l_3}{2l_2}e_3+\frac{l_4}{2l_2}e_4)=0,\ \ g(L, \frac{l_3}{l_2}e_4-\frac{l_4}{l_2}e_3)=0.
\end{split}\end{align}
Hence, if $\Sigma$ is a null hypersurface with normal null vector field $L$, then, the energy $Q[\phi](L,Y)$ on $\Sigma$ only contains the informations of the derivative of $\phi$ which tangent to $\Sigma$.
\end{remark}
\begin{remark}
We have
$$\frac{l_3}{l_2}e_4-\frac{l_4}{l_2}e_3=l_3(\frac{l_4}{l_2}N_1+e_4)
-l_4(\frac{l_3}{l_2}N_1+e_3),$$
and
$$\frac{l_3}{l_2}N_1+e_3,\ \frac{l_4}{l_2}N_1+e_4,\ N_2+\frac{l_3}{2l_2}e_3+\frac{l_4}{2l_2}e_4,$$
are linearly independent. Furthermore, we have
\be
L=l_2(N_2+\frac{l_3}{2l_2}e_3+\frac{l_4}{2l_2}e_4)+\frac{l_3}{2}(\frac{l_3}{l_2}N_1+e_3)
+\frac{l_4}{2}(\frac{l_4}{l_2}N_1+e_4).
\ee
\end{remark}
\begin{remark}
We consider the boundary $\mathcal{H}$ of $\Omega_{T,0}$, then, the null normal vector field $N_\mathcal{H}$ is given by
$$N_\mathcal{H}=l_1N_1+l_2N_2+l_3e_3+l_4e_4,$$
with $l_1\geq0$, $l_2\geq0$. And, we know, $\mathcal{H}\cap \{\tau=0\}=S_{0,0}$, and on
$S_{0,0}$, there satisfies $N_\mathcal{H}=N_2$, i.e. $l_2=1$ and $l_1=l_3=l_4=0$.
\end{remark}
\begin{remark}
We assume
$$g(N_\mathcal{H},N_1)=-l_2\neq0,$$
satisfied on $\mathcal{H}$, which means, no where on $\mathcal{H}$ the null hypersurface $\{\tau=const\}$ will tangent to $\mathcal{H}$.
\end{remark}
\begin{remark}
We denote the three tangential vector fields of $\mathcal{H}$ by
$$E_1=\frac{l_3}{l_2}N_1+e_3,\ \ E_2=\frac{l_4}{l_2}N_1+e_4,\ E_3=N_2+\frac{l_3}{2l_2}e_3+\frac{l_4}{2l_2}e_4.$$
And, we definite
\begin{align}\begin{split}
&E_{\mathcal{H}}[\phi,b_1,b_2]\\
=&\frac{b_1l_2}{2}[(E_1\phi)^2+(E_2\phi)^2]+b_2l_2(E_3\phi)^2
+\frac{b_2l_2}{4}(l_3E_2\phi-l_4E_1\phi)^2+(b_2l_1+b_1l_2)\phi^2,
\end{split}\end{align}
which relate to the energy on $\mathcal{H}$.
\end{remark}
\subsection{$H^1$-estimates in $\mathcal{M}_1$}
In this subsection, we will get the energy estimates in $\mathcal{M}_1$. We use the method which has been used in \cite{LinZhang2017}. For completeness, we will state some result in \cite{LinZhang2017} and give simple proof.

First, we make the $H^1$-estimate in $\Omega_{T}\cap \mathcal{M}_1$, i.e $\Omega_{T,0}$. A simple foliation of this region by the null hypersurfaces $\tau=const$,  but the energy on null hypersurface don't contain the transversal derivative information. In \cite{LinZhang2017}, a weighted energy estimate method has been used, for which, the domain integral control all derivatives, and making suitable choose such that the boundary integrals have good signs.

For the needs of higher derivative estimates, we consider more general wave equations
\be\label{general-wave-eq}
\mathcal{L}v=\Box_gv+a^i\partial_iv+\omega v=f.
\ee
And, for any $C^1$-function $h$, and vector field $X$, we have
\begin{align}\label{divf}
\mathrm{div}(hP[\phi,X])
=(\Box_{g}\phi)(hX\phi)+\frac{1}{2}hQ[\phi]_{\alpha\beta}\ ^{(X)}\pi^{\alpha\beta}
+Q[\phi](\nabla h, X)
-2g(X, \nabla\phi)h\phi,
\end{align}
where
\begin{equation*}
P[\phi,X]_\alpha=Q[\phi]_{\alpha\beta}X^\beta,\ \ ^{(X)}\pi^{\alpha\beta}=\partial^{\alpha}X^{\beta}
+\partial^{\beta}X^{\alpha}-X(g^{\alpha\beta}).
\end{equation*}
\begin{theorem}\label{main-theorem}
For some fixed $T>0$, and $v$ satisfies \eqref{general-wave-eq}. Then, there exists constants $q_0>0, l>0$ depending on $|g^{ij}|_{C^1(\Omega_T)}$, $|a^i|_{L^\infty(\Omega_T)}$ and $|\omega|_{L^\infty(\Omega_T)}$, such that, for any
$q>q_0$, and $h=e^{-lq\tau+qz}$, \footnote{For a hypersurface $\Sigma$, we can define the $H^p(\Sigma)$ space and the corresponding $H^p(\Sigma)$-norm, with derivatives taken only with respect to variables on $\Sigma$. And, let $h$ the positive function, define
$$||u||_{H_h^p(\Omega)}=\sum_{i=0}^{p}||h^{\frac{1}{2}}\nabla^iu||_{L^2(\Omega)},$$
and similarly for $||u||_{H_h^p(\Sigma)}$.}
\begin{align}\begin{split}
&q^{\frac{1}{2}}||v||_{H_h^1(\Omega_{T,0})}+||\partial_z v||_{L_h^2(\mathcal{T})}+||v||_{H_h^1(\Sigma_T)}
+\int_{\mathcal{H}}hE_{\mathcal{H}}[v,1,1]\text{d}\mathcal{H}\\
\leq&\ C\{||v||_{H_h^1(\Sigma_0)}+||v||_{H_h^1(\mathcal{T})}+||f||_{L_h^2(\Omega_{T,0})}\},
\end{split}\end{align}
where $C$ is a constant depending on $|g^{ij}|_{C^1(\Omega_T)}$, $|a^i|_{L^\infty(\Omega_T)}$ and $|\omega|_{L^\infty(\Omega_T)}$.
\end{theorem}
\begin{proof}
For some constants $m,p,q>0$ to be chosen later, we consider a timelike vector field
\begin{equation}\label{eq-definition-Y}Y=mN_1+N_2,\end{equation} and functions
\begin{equation}\label{eq-definition-h}w=-(p\tau-qz),\quad h= e^{w}.\end{equation}
Then, an integration over $\Omega_{T,0}$ yields
\begin{equation}\label{sf}
\int_{\Omega_{T,0}}\mathrm{div}(hP)\  d\Omega=\int_{\partial\Omega_{T,0}}i_{hP}\ d\Omega,
\end{equation}
where $d\Omega$ is  the volume element.
 By (\ref{divf}), and (\ref{sf}), we obtain
\begin{align}\label{keyf}\begin{split}
&\int_{\Omega_{T,0}}\{Q[v](\nabla h, Y)
+\frac{1}{2}hQ[v]_{\alpha\beta}\ ^{(Y)}\pi^{\alpha\beta}-2g(Y, \nabla v)hv\\
&\quad\qquad\qquad +(\mathcal{L}v-a^{\alpha}\partial_{\alpha}v-\omega v)h(Yv)\}d\Omega\\
&\qquad= \int_{\mathcal{M}_1\cap\Sigma_T}hQ[v](\nabla\tau, Y)\sqrt{\gamma}dzd\Sigma
+\int_{\mathcal{H}}hE_\mathcal{H}[v,m,1] \text{d}\mathcal{H}\\
&\qquad\quad -\int_{\Sigma_0}hQ[v](\nabla\tau, Y)\sqrt{\gamma}dzd\Sigma
+\int_{\mathcal{T}}hQ[v](\nabla z, Y)\sqrt{\gamma}d\tau d\Sigma.
\end{split}\end{align}
where $\text{d}\mathcal{H}$ is the volume element for $\mathcal{H}$ adapting to $N_\mathcal{H}$.
\begin{remark}
By our choosing of $h$, we have
$$\nabla h=-h[(p-\frac{1}{2}z^2Vq)N_1+qN_2].$$
We choose $p,q$ such that $\nabla h$ is timelike, hence, $Q[v](\nabla h,Y)$ control all derivatives of $v$, and can absorb the other terms in the domain integral.
\end{remark}
In fact
\begin{align}\label{1}\begin{split}
&-Q[v](\nabla w, Y)= Q[v]([p-\frac{1}{2}z^{2}Vq]N_1+qN_{2}, mN_1+N_{2})\\
= & q(\partial_{\tau}v)^{2}+[mp+\frac{1}{2}qz^{2}V(\frac{1}{2}z^{2}V-m)](\partial_{z}v)^{2}
 +(mq+p-\frac{1}{2}qz^{2}V)[\frac12g^{AB}\partial_{A}v\partial_{B}v+v^{2}]\\
&\quad +qz^{2}V\partial_{\tau}v\partial_{z}v+2qg^{1A}\partial_{A}v\partial_{\tau}v
+qz^{2}Vg^{1A}\partial_{A}v\partial_{z}v+q(g^{1A}\partial_{A}v)^{2}.
\end{split}\end{align}
We consider $p=lq$, for some constant $l>0$ sufficiently large depending on $m$ and $|g^{ij}|_{L^\infty(\Omega_T)}$, then
\begin{align*}
-Q[v](\nabla w, Y)\geq  \frac{1}{2}q\{(\partial_{\tau}v)^{2}+(\partial_{z}v)^{2}
+g^{AB}\partial_{A}v\partial_{B}v+v^{2}\}.
\end{align*}
Furthermore, by choosing $q$
sufficiently large, depending on $|g^{ij}|_{C^{1}(\Omega_{T})}$,
$|a^{\alpha}|_{L^{\infty}(\Omega_{T})}$, and $|\omega|_{L^{\infty}(\Omega_{T})}$, we obtain
\begin{align}\label{eq-estimate-domain-integral}\begin{split}
& \int_{\Omega_{T,0}}\big\{Q[v](\nabla h, Y)
+\frac{1}{2}hQ[v]_{\alpha\beta}\ ^{(Y)}\pi^{\alpha\beta}-2g(Y, \nabla v)hv\\
 &\qquad\quad+(\mathcal{L}v-a^{\alpha}\nabla_{\alpha}v-\omega v)h(Yv)\big\}d\Omega\\
&\quad \leq -\frac{1}{4}\int_{\Omega_{T,0}}h\{q[(\partial_{\tau}v)^{2}
+(\partial_{z}v)^{2}+g^{AB}\partial_{A}v\partial_{B}v+v^{2}]-(\mathcal{L}v)^{2}\}d\Omega.
\end{split}\end{align}
Now we analyze the boundary integrals. By simple calculation, we have
\begin{equation}\label{3}
Q[v](\nabla\tau, Y)=Q[v](N_1,mN_1+N_2)= m(\partial_{z}v)^{2}+\frac{1}{2}g^{AB}\partial_{A}v\partial_{B}v+v^{2},
\end{equation}
and
\begin{equation}\label{3a}\begin{split}
&Q[v](\nabla z, Y)=Q[v](-N_2+\frac{1}{2}z^2VN_1,mN_1+N_2)\\ =&-(N_2v)^{2}+\frac{1}{2}z^2Vm(\partial_{z}v)^{2}
-(m-\frac{1}{2}z^2V)[\frac{1}{2}g^{AB}\partial_{A}v\partial_{B}v+v^{2}]\\
=&-(\partial_{\tau}v+g^{1A}\partial_{A}v)^{2}
- (m-\frac12z^2V)(\frac{1}{2}g^{AB}\partial_{A}v\partial_{B}v+v^{2})\\
&\qquad+\frac12z^2V(m-\frac{1}{2}z^{2}V)(\partial_{z}v)^{2}
-z^{2}V\partial_{\tau}v\partial_{z}v
-z^{2}Vg^{1A}\partial_{A}v\partial_{z}v.
\end{split}\end{equation}
\begin{remark}
We know the boundary hypersurface $\mathcal{T}$ is timelike, hence the energy $Q[v](\nabla z,Y)$ is not positive on $\mathcal{T}$. As our initial data setting on this boundary, hence we know $v$, $\partial_\tau v$, and $\nabla_{S^2}v$, but we does not yield any information on $\partial_zv|_\mathcal{T}$. We choose $m$, such that the term $(\partial _z v)^2$ in \eqref{3a} has a good sign, so we can control this boundary term by initial data.
\end{remark}
By choosing $m$ large, such that $m-\frac{1}{2}z^2V\geq 1$ on $\mathcal{T}\cup \{\tau\leq T\}$, we get
\begin{align*}
& q\int_{\Omega_{T,0}}h\left[(\partial_{\tau}v)^{2}+(\partial_{z}v)^{2}
+g^{AB}\partial_{A}v\partial_{B}v+v^{2}\right]d\Omega
+\int_{\mathcal{T}}h(\partial_{z}v)^{2}d\tau d\Sigma\\
&+\int_{\Sigma_T\cap\mathcal{M}_1}h[(\partial_zv)^2
+g^{AB}\partial_Av\partial_Bv+v^2]dzd\Sigma+\int_{\mathcal{H}}hE_\mathcal{H}[v,1,1] \text{d}\mathcal{H}\\
\leq&  C\{\int_{\Omega_{T,0}}hf^{2}d\Omega+
\int_{\Sigma_0}h[(\partial_zv)^2
+g^{AB}\partial_Av\partial_Bv+v^2]dzd\Sigma\\
&\qquad+\int_{\mathcal{T}}h[(\partial_\tau v)^2
+g^{AB}\partial_Av\partial_Bv+v^2]d\tau d\Sigma\},
\end{align*}
where $C$ is a constant depending on $|g^{ij}|_{C^1(\Omega_T)}$, $|a^i|_{L^\infty(\Omega_T)}$ and $|\omega|_{L^\infty(\Omega_T)}$.
\end{proof}
\begin{remark}
 For the boundary integral on $\mathcal{T}$, we take a  small $m$, such that $m-\frac{1}{2}z^2V<0$ on $\mathcal{T}$, hence, the sign of the term $\frac{1}{2}g^{AB}\partial_{A}v\partial_{B}v+v^{2}$ in \eqref{3a} is positive, then we have another version of the $H^1$-estimate,
\begin{align}\begin{split}
&q^{\frac{1}{2}}||v||_{H_h^1(\Omega_{T,0})}+||\partial_z v||_{L_h^2(\mathcal{T})}+||v||_{H_h^1(\Sigma_T)}
+\int_{\mathcal{H}}hE_{\mathcal{H}}[v,1,1]\text{d}\mathcal{H}\\
\leq&\ C\{||v||_{H_h^1(\Sigma_0)}+||N_2v||_{L_h^2(\mathcal{T})}+||f||_{L_h^2(\Omega_{T,0})}\}.
\end{split}\end{align}
\end{remark}
\subsection{Higher order estimates in $\mathcal{M}_1$}
In this subsection, we derive $H^k$-estimates, for $k\geq2$. To this end, we need to differentiate the equations \eqref{wave-equ}, In fact, let $X$ be a vector  field with the deformation tensor $\pi=\ ^{(X)}\pi$, then (see \cite{Alinhac}),
\begin{equation*}
[\Box,X]\phi=\pi^{\alpha\beta}\nabla^{2}\phi_{\alpha\beta}
+\nabla_{\alpha}\pi^{\alpha\beta}\partial_{\beta}\phi
-\frac{1}{2}\partial^{\alpha}(tr\pi)\partial_{\alpha}\phi.
\end{equation*}
In the following, we denote the multi-indices $\alpha,\beta\in \mathbb{Z}_+^4$ by
$\alpha=(\alpha_0,\alpha_1,\alpha_2,\alpha_3)$, $\beta=(\beta_0,\beta_1,\beta_2,\beta_3)$, etc. Take an arbitrary multi-index $\alpha$ with $|\alpha|=p$, then
\be\label{higher-wave-equ}
\Box_g(\partial^\alpha v)=\sum_{|\beta|=p+1}c_{\alpha\beta}\partial^\beta v+f_\alpha,
\ee
where
\be
f_\alpha=\sum_{|\beta|\leq p}c_{\alpha\beta}\partial^\beta v.
\ee
\begin{remark}
We will make $H^1$-estimate of  \eqref{higher-wave-equ}, there are two issues we need to resolve. First, the right hand  contains derivatives of $v$ of order $|\alpha|+1$,
not all of which can be written as $\partial_i\partial^\alpha v$, for some $i=0, 1, 2, 3$. If we
simply apply the derived $H^1$-estimates, there are derivatives of order $|\alpha|+1$ in the right-hand side,
which are not yet controlled. Second, we need to determine
the initial values of $\partial^\alpha v$ on $\Sigma_0$ and
the boundary values on $\Sigma_1$.
\end{remark}
\begin{theorem}\label{main-theorem-high} For any $p\geq1$, there exists $q_0$ depending on $|g^{ij}|_{C^1(\Omega_T)}$, $|a^i|_{L^\infty(\Omega_T)}$, $|\omega|_{L^\infty(\Omega_T)}$ and $p$, and  $l$  the same  as in theorem \ref{main-theorem}, then, for $q\geq q_0$ and $h=e^{-ql\tau+qz}$, we have
\begin{align}\begin{split}
&\sum_{|\alpha|=p+1}(q^{\frac{1}{2}}||\partial^\alpha v||_{L_h^2(\Omega_{T,0})}+||\partial^\alpha v||_{L_h^2(\mathcal{T})}
+\int_{\mathcal{H}}hE_{\mathcal{H}}[\partial^\alpha v,1,1]\text{d}\mathcal{H})\\
\leq&\ C\{||\varphi||_{H_h^{(2p+2)}(\Sigma_0)}+||\psi_1||_{H_h^{2p+2}(\mathcal{T})}
+||f||_{H_h^{2p+1}(\Omega_{T,0})}\},
\end{split}\end{align}
where $C$ is constant depending on
\end{theorem}
\begin{proof}
By Theorem \ref{main-theorem}, and \eqref{higher-wave-equ}, we have, for $q>q_0$, and $h=e^{-ql\tau+qz}$,
\begin{align}\label{ineq-1}\begin{split}
&q^{\frac{1}{2}}||\partial^\alpha v||_{H_h^1(\Omega_{T,0})}+||\partial_z \partial^\alpha v||_{L_h^2(\mathcal{T})}+||\partial^\alpha v||_{H_h^1(\Sigma_T)}
+\int_{\mathcal{H}}hE_{\mathcal{H}}[\partial^\alpha v,1,1]\text{d}\mathcal{H}\\
\leq&\ C\{||\partial^\alpha v||_{H_h^1(\Sigma_0)}+||\partial^\alpha v||_{H_h^1(\mathcal{T})}+\sum_{|\beta|=p+1}||\partial^\beta v||_{L_h^2(\Omega_{T,0})}
+||f_\alpha||_{L_h^2(\Omega_{T,0})}\},
\end{split}\end{align}
where $C$ is constant depending on $|g^{ij}|_{C^2(\Omega_T)}$, $|a^i|_{C^1(\Omega_T)}$ and $|\omega|_{L^\infty(\Omega_T)}$.
Specially, from \eqref{ineq-1}, we have, for $1\leq r\leq p+1$,
\begin{align}\label{ineq-2}\begin{split}
&\sum_{|\alpha|=p+1,\alpha_1=r}||\partial^\alpha v||_{L_h^2(\mathcal{T})}\\
\leq&\ C\{\sum_{|\alpha|\leq p+1}||\partial^\alpha v||_{L_h^2(\Sigma_0)}+\sum_{|\alpha|=p+1,\alpha_1=r-1}||\partial^\alpha v||_{L_h^2(\mathcal{T})}+\sum_{|\alpha|= p}||\partial^\alpha v||_{L_h^2(\mathcal{T})}\\
&\qquad+\sum_{|\beta|=p+1}||\partial^\beta v||_{L_h^2(\Omega_{T,0})}
+\sum_{|\alpha|=p+1}||f_\alpha||_{L_h^2(\Omega_{T,0})}\}.
\end{split}\end{align}
Hence, we have, for any $0\leq r\leq p+1$,
\begin{align}\label{ineq-3}\begin{split}
&\sum_{|\alpha|=p+1,\alpha_1=r}||\partial^\alpha v||_{L_h^2(\mathcal{T})}\\
\leq&\ C\{\sum_{|\alpha|\leq p+1}||\partial^\alpha v||_{L_h^2(\Sigma_0)}+\sum_{|\alpha|=p+1,\alpha_1=0}||\partial^\alpha v||_{L_h^2(\mathcal{T})}+\sum_{|\alpha|= p}||\partial^\alpha v||_{L_h^2(\mathcal{T})}\\
&\qquad+\sum_{|\beta|=p+1}||\partial^\beta v||_{L_h^2(\Omega_{T,0})}
+\sum_{|\alpha|=p+1}||f_\alpha||_{L_h^2(\Omega_{T,0})}\}.
\end{split}\end{align}
By substituting \eqref{ineq-3} in \eqref{ineq-1}, we have, for $q$ sufficiently large,
\begin{align}\label{ineq-4}\begin{split}
&\sum_{|\alpha|=p+1}(q^{\frac{1}{2}}||\partial^\alpha v||_{L_h^2(\Omega_{T,0})}+||\partial^\alpha v||_{L_h^2(\mathcal{T})}
+\int_{\mathcal{H}}hE_{\mathcal{H}}[\partial^\alpha v,1,1]\text{d}\mathcal{H})\\
\leq&\ C\{\sum_{|\alpha|\leq p+1}||\partial^\alpha v||_{L_h^2(\Sigma_0)}+\sum_{|\alpha|=p+1,\alpha_1=0}||\partial^\alpha v||_{L_h^2(\mathcal{T})}\\
&\qquad+\sum_{|\alpha|= p}||\partial^\alpha v||_{L_h^2(\mathcal{T})}
+\sum_{|\alpha|=p+1}||f_\alpha||_{L_h^2(\Omega_{T,0})}\},
\end{split}\end{align}
the last two terms on the right hand can be controlled by lower order estimates,  the second term yield by $\psi_1$, and by the following estimates on $\Sigma_0$, we can get the desired result.
\end{proof}
To control the first term in the right hand of \eqref{ineq-4}, we need following Lemma, which also can be seen in \cite{LinZhang2017}.
\begin{lemma}
For any $\alpha$ with $|\alpha|=p$, and $0\leq\alpha_0=l\leq p$, we have
\begin{align}\label{initdata}
\|\partial^{\alpha}v\|_{L^{2}(\Sigma_{0})}
\leq C\big\{\|f\|_{H^{p+l-1}(\Omega_T)}+\|\varphi\|_{H^{p+l}(\Sigma_{0})}
+\|\psi_1\|_{H^{p+l}(\Sigma_{1})}\big\},
\end{align}
where $C$ is a positive constant depending only on $p$, $|g^{ij}|_{C^{p+l-1}(\Sigma_{0})}$,
$|a^{i}|_{C^{p+l-2}(\Sigma_{0})}$, and $|\omega|_{C^{p+l-2}(\Sigma_{0})}$.
\end{lemma}
\begin{proof}
By restricting the equation $\mathcal{L}u=f$ to $\Sigma_0$, we have
\begin{equation}\label{eq-ODE-tau}
2\partial_{z}(\partial_{\tau}v)+a^{0}\partial_\tau v=f_1,
\end{equation}
where
$$f_1=f+\sum_{|\beta|\le 2, \beta_0=0}a_\beta\partial^{\beta}v.$$
We point out that no derivatives of $v$ with respect to $\tau$ appear in $f_1$.
We view \eqref{eq-ODE-tau} as an ODE of $\partial_\tau v$ in $z$ on $\Sigma_0$ with the initial value given by
$$\partial_\tau v=\partial_\tau \psi_1\quad\text{on }\Sigma_0\cap \Sigma_1.$$
Then,
\begin{align*}
\partial_\tau v=\partial_\tau \psi e^{-\frac{1}{2}\int_{z}^{z_{0}}a^{0}dz'}
-\frac12\int_z^{z_0}f_1e^{\frac{1}{2}\int_{z'}^{z_{0}}a^{0}dz''}dz'
\quad\text{on }\Sigma_0.\end{align*}
Therefore,
\begin{equation}\label{initdata-1z}
\|\partial_\tau v\|_{L^{2}(\Sigma_{0})}
\leq C\left\{\|\varphi\|_{H^{2}(\Sigma_{0})}
+\|\partial_\tau \psi_1\|_{L^2(\Sigma_{0}\cap \Sigma_{1})}+\|f\|_{L^2(\Sigma_{0})}\right\}.
\end{equation}
For $l\ge 2$, by applying $\partial_{\tau}^{l-1}$ to \eqref{eq-ODE-tau}, we obtain
\begin{equation*}
2\partial_{z}(\partial_{\tau}^{l}v)+a\partial_{\tau}^{l}v=f_{l},
\end{equation*}
where
$$f_{l}=\partial_{\tau}^{l-1}f
+\sum_{|\beta|\le l+1, \beta_0\le l-1}c_{\beta}\partial^{\beta}v.$$
Similarly, we view this as an ODE of $\partial_{\tau}^{l}v$ in $z$ on $\Sigma_0$
with the initial value given  by
$$\partial_{\tau}^{l}v=\partial_{\tau}^{l}\psi_1\quad\text{on }\Sigma_0\cap \Sigma_1.$$
Then,
\begin{align*}
\partial^{l}_\tau v=\partial^{l}_\tau \psi_1 e^{-\frac{1}{2}\int_{z}^{z_{0}}adz'}
-\frac12\int_z^{z_0}f_{l}e^{\frac{1}{2}\int_{z'}^{z_{0}}adz''}dz'
\quad\text{on }\Sigma_0.\end{align*}
For
$\alpha=(l, \alpha_1, \alpha_2, \alpha_3)$, we write $\alpha'=(0, \alpha_1, \alpha_2, \alpha_3)$. Then,
\begin{align*}
\partial^\alpha v=
\partial^{\alpha'}\big\{\partial^{l}_\tau \psi_1 e^{-\frac{1}{2}\int_{z}^{z_{0}}adz'}
-\frac12\int_z^{z_0}f_{l}e^{-\frac12\int_{z'}^{z_{0}}adz''}dz'\big\}
\quad\text{on }\Sigma_0.\end{align*}
Hence, for $\alpha$ with $|\alpha|=p$ and $\alpha_0=l$,
\begin{align*}
\|\partial^{\alpha}v\|_{L^{2}(\Sigma_{0})}
\leq C\big\{\|f\|_{H^{p-1}(\Sigma_{0})}
+\|\partial_\tau^{l}\psi_1\|_{H^{p-l}(\Sigma_{0}\cap \Sigma_{1})}
+\sum_{|\beta|\le p+1, \beta_0\le l-1}\|\partial^{\beta}v\|_{L^{2}(\Sigma_{0})}
\big\}.
\end{align*}
By the trace theorem,
we have
\begin{align*}
\|\partial^{\alpha}v\|_{L^{2}(\Sigma_{0})}
\leq C\big\{\|f\|_{H^{p}(\Omega)}
+\|\psi_1\|_{H^{p+1}(\Sigma_{1})}
+\sum_{|\beta|\le p+1, \beta_0\le l-1}\|\partial^{\beta}v\|_{L^{2}(\Sigma_{0})}
\big\}.
\end{align*}
We note that in the summation above, the highest degree of derivatives increases by 1 but the highest
degree of derivatives with respect to $\tau$ decreases by 1.
 So we can iterate this inequality $l$ times and obtain the desired result.
\end{proof}
\begin{remark}
We point out that there is a loss of regularity which caused by to control the $\tau$-direction derivatives on $\Sigma_0$. In fact, in \cite{LinZhang2017} has proved, for $|\alpha|=k$, and $\alpha_0=l\leq k$, then
\be
||\partial^\alpha v||_{L^2(\mathcal{M}_1\cap\Omega_T)}\leq C\{||\varphi||_{H^{k+l}(\Sigma_0)}+||\psi_1||_{H^{k+l}(\mathcal{T})}\}.
\ee
\end{remark}
\begin{corollary}\label{M-1-coroll}
we have
\begin{align}\begin{split}
||v||_{H^k(\mathcal{M}_1\cap\Omega_T)}\leq C\{||\varphi||_{H^{2k}(\Sigma_0)}+||\psi_1||_{H^{2k}(\mathcal{T})}\},
\end{split}\end{align}
where C is a constant depending on $|g^{ij}|_{C^{2k-1}(\Omega_T)}$, $|\omega|_{C^{2k-2}(\Omega_T)}$, $k$, $T$, and $R$.
\end{corollary}
\subsection{Energy estimates in $\mathcal{M}_2$}
In this subsection, we will get the energy estimates in $\mathcal{M}_2\cap\Omega_T$. We know the boundary of $\mathcal{M}_2\cap\Omega_T$ are $\Sigma_T\cap\mathcal{M}_2$, $\mathcal{H}$ and $\mathcal{I}$. In fact, the energy estimates are same as the last subsection.
\begin{theorem}\label{M-2-main-theorem}There exists $q_0$ and $l$ depending on $|g^{ij}|_{C^1(\mathcal{M}_2\cap\Omega_T)}$, $|\omega|_{L^\infty(\mathcal{M}_2\cap\Omega_T)}$,  such that, for any $q\geq q_0$, and $h=e^{-ql\tau+qz}$,
\begin{align}\begin{split}
&q^{\frac{1}{2}}||v||_{H_h^1(\Omega_T\cap\mathcal{M}_2)}
+||v||_{H_h^1(\Sigma_T\cap\mathcal{M}_2)}+||\partial_zv||_{L_h^2(\mathcal{I})}\\
\leq& C\{\int_{\mathcal{H}}hE_{\mathcal{H}}[v,1,1]\text{d}\mathcal{H}+||v||_{H_h^1(\mathcal{I})}
+||f||_{L_h^2(\mathcal{M}_2\cap\Omega_T)}\},
\end{split}\end{align}
where $C$ is a constant depending on $|g^{ij}|_{C^1(\mathcal{M}_2\cap\Omega_T)}$, $|\omega|_{L^\infty(\mathcal{M}_2\cap\Omega_T)}$ and $\Lambda$.
\end{theorem}
\begin{proof} The prove is similar to Theorem \ref{main-theorem}. We also take $Y=mN_1+N_2$ and $h=e^{-ql\tau+qz}$, then
\begin{align}\label{keyf-2}\begin{split}
&\int_{\mathcal{M}_2\cap \Omega_T}\{Q[v](\nabla h, Y)
+\frac{1}{2}hQ[v]_{\alpha\beta}\ ^{(Y)}\pi^{\alpha\beta}-2g(Y, \nabla v)hv\\
&\quad\qquad\qquad +(\mathcal{L}v-a^{\alpha}\partial_{\alpha}v-\omega v)h(Yv)\}d\Omega\\
=& \int_{\mathcal{M}_2\cap\Sigma_T}hQ[v](\nabla\tau, Y)\sqrt{\gamma}dzd\Sigma
-\int_{\mathcal{H}}hE_\mathcal{H}[v,m,1] \text{d}\mathcal{H}
-\int_{\mathcal{I}}hQ[v](\nabla z, Y)\sqrt{\gamma}d\tau d\Sigma.
\end{split}\end{align}
The domain integral and the boundary term on $\Sigma_T\cap\mathcal{M}_2$  can be analysis by the same way in Theorem \ref{main-theorem}, and the boundary term on $\mathcal{H}$ has controlled by Theorem \ref{main-theorem-high}.  Hence, we only need to analysis the boundary term on $\mathcal{I}$. In fact, on $\mathcal{I}$, $N_2=\partial_\tau-\frac{\Lambda}{6}\partial_z$, and then,
\begin{align}\begin{split}
-Q[v](\nabla z, Y)&=(N_2v)^{2}+\frac{\Lambda}{6}m(\partial_{z}v)^{2}
+(m+\frac{\Lambda}{6})[\frac{1}{2}g^{AB}\partial_{A}v\partial_{B}v+v^{2}]\\
&=(\partial_\tau v)^2+\frac{\Lambda}{6}(\frac{\Lambda}{6}+m)(\partial_z v)^2
+(m+\frac{\Lambda}{6})[\frac{1}{2}g^{AB}\partial_{A}v\partial_{B}v+v^{2}]
-\frac{\Lambda}{3}\partial_\tau v\partial_z v\\
&\geq -7(\partial_\tau v)^2+\frac{\Lambda}{6}(\frac{\Lambda}{12}+m)(\partial_z v)^2
+(m+\frac{\Lambda}{6})[\frac{1}{2}g^{AB}\partial_{A}v\partial_{B}v+v^{2}].
\end{split}\end{align}
Hence, we take $m=-\frac{\Lambda}{24}$, then we can get the desired result.
\end{proof}

From Theorem \ref{M-2-main-theorem}, and by the method in Theorem \ref{main-theorem-high}, we have
\begin{theorem}\label{M-2-mian-theorem-high}
For any $\alpha$ with $|\alpha|=k\geq2$, there exists $q_0$ depending on $k$, $|g^{ij}|_{C^1(\mathcal{M}_2\cap\Omega_T)}$, $|\omega|_{L^\infty(\mathcal{M}_2\cap\Omega_T)}$, and $l$ depending on $|g^{ij}|_{C^1(\mathcal{M}_2\cap\Omega_T)}$, $|\omega|_{L^\infty(\mathcal{M}_2\cap\Omega_T)}$, such that, for any $q\geq q_0$, and $h=e^{-ql\tau+qz}$,
\begin{align}\begin{split}
&\sum_{|\alpha|=k}[q^{\frac{1}{2}}||\partial^\alpha v||_{H_h^1(\Omega_T\cap\mathcal{M}_2)}
+||\partial^\alpha v||_{L_h^2(\mathcal{I})}]\\
\leq& C\{\sum_{|\alpha|\leq k}\int_{\mathcal{H}}hE_{\mathcal{H}}[\partial^\alpha v,1,1]\text{d}\mathcal{H}
+||\psi_2||_{H_h^k(\mathcal{I})}+||f||_{H_h^{k-1}(\mathcal{M}_2\cap\Omega_T)}\},
\end{split}\end{align}
where $C$ is a constant depending on $k$, $|g^{ij}|_{C^k(\mathcal{M}_2\cap\Omega_T)}$, $|\omega|_{C^{k-1}(\mathcal{M}_2\cap\Omega_T)}$ and $\Lambda$.
\end{theorem}
Combine Theorem \ref{main-theorem-high} and Theorem \ref{M-2-mian-theorem-high}, we have
\begin{corollary}
Suppose $v$ is a solution of \eqref{wave-equ}-\eqref{wave-boundary}, then, for any $k\geq1$, we have
\begin{align}\begin{split}
||v||_{H^k(\Omega_T)}\leq C\{||\varphi||_{H^{2k}(\Sigma_0)}+||\psi_1||_{H^{2k}(\mathcal{T})}
+||\psi_2||_{H^k(\mathcal{I})}\},
\end{split}\end{align}
where C is a constant depending on $|g^{ij}|_{C^{2k-1}(\Omega_T)}$, $|\omega|_{C^{2k-2}(\Omega_T)}$, $k$, $\Lambda$, $T$, and $R$.
\end{corollary}
\begin{proof}
We have used $e^{-qlT}\leq e^{-ql\tau+qz}\leq e^{\frac{q}{R}}$.
\end{proof}
\section{Existence of solutions}
In this section, we will consider the existence of solution of \eqref{wave-equ}-\eqref{wave-boundary}. We first consider the local existence of null-timelike problems.  Just as the same as in \cite{LinZhang2017}, we first proving the existence for the analytic case following the method in \cite{Duff1958}, and then for general case by approximations. For convenience, we consider the existence near $\Sigma_0\cap\mathcal{T}$.
\begin{theorem}[\underline{Local existence of null-timelike problems}]\label{existence-null-timelike}\footnote{The prove same as in \cite{LinZhang2017}}
For linear wave equations \eqref{wave-equ}, with $g^{ij}\in C^{2k-1}$, $\omega\in C^{2k-2}$ and $v|_{\Sigma_0}=\varphi\in H^{2k}(\Sigma_0)$, $v|_{\mathcal{T}}=\psi\in H^{2k}(\mathcal{T})$. Then, there exits a point $O\in \Omega$, and $v\in H^k(\mathcal{J}^-(O)\cap\Omega)$ solve this null-timelike boundary problem.
\end{theorem}
\begin{proof}
We write the equation in \eqref{wave-equ} in the form
\begin{equation}
2\partial_{z\tau}v+z^{2}V\partial_{zz}v+2g^{1A}\partial_{zA}v
+g^{AB}\partial_{AB}v+a^{i}\partial_{i}v+b v=0.
\end{equation}

We first assume that $g^{ij}, a^{i}, b,$ and $\varphi,\psi$ are real analytic in some open neighborhood $\mathcal{U}$ of $\Sigma_0\cap\mathcal{T}$  and hence can be expanded
as a power series in $\tau$. For example, we have
\begin{align*}
\psi=\sum_{i=0}^{\infty}\psi_{i}( \theta)\tau^{i}.
\end{align*}
Here and hereafter, we denote by $\theta=(x^2, x^3)$, coordinates on $S^2$. We define
\begin{align*}
&u^{0}=v=\sum_{i=0}^{\infty} u_{i}^{0}(z, \theta)\tau^{i},\\
&u^{1}=\partial_{z}v=\sum_{i=0}^{\infty} u_{i}^{1}(z, \theta)\tau^{i},\\
&u^{A}=\partial_{A}v=\sum_{i=0}^{\infty} u_{i}^{A}(z, \theta)\tau^{i}, \ A=2,3, \end{align*}
and
\begin{align*}
w=\partial_{\tau}v=\sum_{i=0}^{\infty} w_{i}(z, \theta)\tau^{i}.
\end{align*}
For $u^0, u^1, u^2, u^3$, and $ w$, we have
\begin{align*}
\partial_\tau u^{0}&=w,\\
2\partial_\tau u^{1}&=-z^{2}V\partial_z u^{1}-2g^{1A}\partial_z u^{A}
-g^{AB}\partial_B u^{A}+\bar{a}^{i}u^{i}+cw,\\
\partial_\tau u^{A}&=\partial_A w,\end{align*}
and
\begin{align*}
2\partial_z w=-z^{2}V\partial_z u^{1}-2g^{1A}\partial_z u^{A}
-g^{AB}\partial_B u^{A}+\bar{a}^{i}u^{i}+cw.
\end{align*}
Therefore,
\begin{align}\label{eq-ODE-system-2}\begin{split}
(i+1)u_{i+1}^{0}&=w_{i},\\
2(i+1)u_{i+1}^{1}&=\sum_{k\leq i}
L_{k}[\partial u_{k}^{0}, \partial u_{k}^{1}, \partial u_{k}^{2}, \partial u_{k}^{3}, w_{k}]
,\\
(i+1)u_{i+1}^{A}&=\partial_A w_{i},
\end{split}\end{align}
and
\begin{align}\label{eq-ODE-system-1}
2\partial_z w_{i}+d w_{i}
=\sum_{k\leq i}L_{k}[\partial u_{k}^{0}, \partial u_{k}^{1}, \partial u_{k}^{2}, \partial u_{k}^{3}]
+\sum_{k\leq i-1}L_{k}[w_{k}].\end{align}
Note
$$u_{0}^{0}=\varphi, \quad u_{0}^{1}=\partial_z \varphi, \quad u_{0}^{A}=\partial_A \varphi.$$
The equation \eqref{eq-ODE-system-1} is an ODE of $w_{i}$ with respect to $z$
and the initial value is given by
$$w_{i}(z_0, \cdot)=(i+1)\psi_{i+1}\quad\text{on }S^2.$$
For some $i\ge 0$, assume we already know $u_0^l$, $\cdots$, $u_i^l$, for $l=0, 1, 2, 3$,
and $w_0$, $\cdots$, $w_{i-1}$, then we can find $w_i$ by solving \eqref{eq-ODE-system-1}
and find $u_{i+1}^{l}$, for $l=0, 1, 2, 3$,
by  \eqref{eq-ODE-system-2}.

For simplicity, we assume
$$u_0^0=u_0^1=u_0^2=u_0^3=0,\quad w_0|_{z=z_0}=0.$$
Otherwise, we set
$$\tilde{u}^{0}=u^{0}-\varphi, \ \tilde{u}^{1}=u^{1}-\partial_z \varphi,\
\tilde{u}^{A}=u^{A}-\partial_A \varphi,\ \tilde{w}=w-\partial_\tau\psi.$$

For some $M>0$ and $\rho>0$, we define
$$F(s)=\frac{M}{1-\frac{s}{\rho}}.$$
We now consider a given point on $\Sigma_0$, say $(0,z_{*},0,0)$ with $z_*>0$.
Set
$$s=a^2\tau+a(z-z_{*})+x_2+x_3.$$
In a neighborhood of $(0,z_{*},0,0)$,
take $M>0, \rho>0$, and $a>1$, such that
the function
\begin{equation*}
F(s)
=\frac{M}{1-\rho^{-1}\big(a^2\tau+a(z-z_{*})+x_2+x_3\big)}
\end{equation*}
is a majorizing function of $g^{ij}$, $\bar{a}^{i}$, $c$, $f$ and $\varphi, \psi$. Then,
\begin{align}\label{eq-Majorizing-system}\begin{split}
\partial_\tau u^{i}&=F(s)\{\partial_z u^{1}
+\sum_{A,B=2,3}[\partial_z u^{A}+\partial_B u^{A}
+\partial_A w]+\sum_{j=0}^{3}u^{j}+w+1\},\\
\partial_z w&=F(s)\{\partial_z u^{1}
+\sum_{A,B=2,3}[\partial_z u^{A}+\partial_B u^{A}]+\sum\limits_{j=0}^{3}u^{j}+w+1\}
\end{split}\end{align}
forms a majorizing system. We now treat $s$ as an independent variable.
To construct a special solution $u^{l}=U(s), l=0,1,2,3$, and $w=W(s)$ of \eqref{eq-Majorizing-system},
we consider a system of linear ordinary differential equations given by
\begin{align}\label{eq-Majorizing-ODE}\begin{split}
\{a^2-F(s)(3a+4)\}\frac{d U}{d s}-2F(s)\frac{d W}{d s}&=F(s)(4U+W+1),\\
-F(s)(3a+4)\frac{d U}{d s}+a\frac{d W}{d s}&=F(s)(4U+W+1),
\end{split}\end{align}
with $U(0)=W(0)=0$. Take $\lambda$ small such that
\begin{equation*}
\begin{bmatrix}
a^2-F(s)(3a+4)&-2F(s)\\
-F(s)(3a+4)&a
\end{bmatrix}
\end{equation*}
is positive definite. Then, we can solve \eqref{eq-Majorizing-ODE}
and its solutions $U$ and $W$ are real analytic in the domain of $F(s)$.
Therefore, the domain where $u^{i}$ and $ w$ exist and are real analytic is the same as the domain
where all coefficients and initial values are real analytic. Similarly as in \cite{Duff1958},
by $a^2-F(s)(3a+4)>0$, $-2F(s)<0$, $-F(s)(3a+4)<0$, $a>0$,
the coefficients in the series of $U(s)$ and $W(s)$ are nonnegative, provided that $U(0)$ and $W(0)$ are $0$.
This proves the existence of an analytic solution $v$ in $\mathcal{U}$.

For some point $O\in\Omega$, such that $\mathcal{J}^-(O)\cap\Omega\in\mathcal{U}$, we find analytic sequences of $g_j$ and $\omega_j$, such that $\lim_{j\to\infty}g_j=g$ and $\lim_{j\to\infty}\omega_j=\omega$ uniformly in $C^{2k-1}(\overline{\mathcal{J}^-(O)\cap\Omega})$-norm. And
 we can find sequences of polynomials $P^{j}$ and $Q^{j}$, such that $\lim_{j\to\infty}P^j=\varphi$ in $H^{2k}(\overline{\mathcal{J}^-(O)\cap\Sigma_0})$-norm, and $\lim_{j\to\infty}Q^j=\psi$ in $H^{2k}(\overline{\mathcal{J}^-(O)\cap\mathcal{T}})$-norm.

Denote by $v^{j}$ the solution of $\Box_{g_j}v_j+\omega_jv_j=0$,
with the initial value and the boundary value given by $P^{j}$ and $Q^{j}$, respectively.
By the $H^{k}$-estimates provided by Corollary \ref{M-1-coroll},
we find that the $v^{j}$ converges, as $j\rightarrow\infty$, to a solution
$v\in H^k(\mathcal{J}^-(O)\cap\Omega)$ of \eqref{wave-equ}
with the initial value and the boundary value given by $\varphi$ and $ \psi$, respectively.
\end{proof}
We also need the local existence theorem for characteristic initial value problem for hyperbolic equations by Rendall \cite{Rendall1990},
\begin{theorem}[\underline{Local existence of characteristic initial value problem}]\label{existence-chara}
Let $\mathcal{N}_1$, $\mathcal{N}_2$ be the transversely intersecting null hypersurfaces with respect to $g$, $\Omega$ be the region bounded by $\mathcal{N}_1$, $\mathcal{N}_2$.
Let $v|_{\mathcal{N}_1}=\varphi$ smooth on $\mathcal{N}_1$,  $v|_{\mathcal{N}_2}=\psi$ smooth on $\mathcal{N}_2$, and $\varphi=\psi$ on $\mathcal{N}_1\cap\mathcal{N}_2$. Then there exists an open neighbourhood $\mathcal{U}$ of $\mathcal{N}_1\cap\mathcal{N}_2$, and a unique $v\in C^\infty(\mathcal{U}\cap\Omega)$ solve \eqref{wave-equ} and satisfies the boundary condition.
\end{theorem}

We now prove the existence of solution in $\mathcal{M}_1$. We consider the spacelike slice
$$S_t=\{(\tau,z,x_2,x_3)|\arctan(\tau+\frac{1}{z})-\frac{\pi}{2}=t\}\cap \Omega_T.$$
\begin{remark}
From the prove of $H^k$ estimates in $\mathcal{M}_1$, we only need to modify the integral domain, which bounded by $S_t$, $\Sigma_T$ and the boundary of $\mathcal{M}_1$, hence we have
\be
\sum_{|\alpha|=k}||\partial^\alpha v||_{L^2(S_t)}\leq C\{||\varphi||_{H^{2k}(\Sigma_0)}+||\psi_1||_{H^{2k}(\mathcal{T})}\},
\ee
where C is a constant depending on $|g^{ij}|_{C^{2k-1}(\Omega_T)}$, $|\omega|_{C^{2k-2}(\Omega_T)}$, $k$, $T$, and $R$.
\end{remark}

\begin{theorem}[\underline{Existence of solution in $\mathcal{M}_1\cap\Omega_T$}]
For $\varphi\in H^{2k}(\Sigma_0)$, and $\psi_1\in H^{2k}(\mathcal{T})$, then, there exists $v\in H^k(\mathcal{M}_1\cap\Omega_T)$ solving $\eqref{wave-equ}$.
\end{theorem}
\begin{proof}
First by Theorem \ref{existence-null-timelike}, we have local existence. If the solution does not exist in the full region $\mathcal{M}_1\cap\Omega_T$, we must have $t_\star$, such that
$$t_\star=\sup\{t:\text{the solution exists in}\ \mathcal{M}_1\cap\Omega_T\cap\cup_{l<t}S_l\}.$$
We take $\epsilon$ sufficiently small, then solution exists in $\mathcal{M}_1\cap\Omega_T\cap \cup_{l\leq t_\star-\epsilon}S_l$, and
\be
\sum_{|\alpha|=k}||\partial^\alpha v||_{L^2(S_{t_\star-\epsilon})}\leq C\{||\varphi||_{H^{2k}(\Sigma_0)}+||\psi_1||_{H^{2k}(\mathcal{T})}\},
\ee
where $C$ is not depend on $t_\star-\epsilon$.
\begin{figure}
\begin{center}
\begin{tikzpicture}
\draw[thin]
(0,3.5)--(0,0);
\node at (-0.2,3) {$\mathcal{T}$};
\draw[thin]
(0,0)--(4,4);
\node at (2.35,2) {$\Sigma_0$};
\draw[thin]
(4,4)--(4,6.5);
\node at (4.2,5.7) {$\mathcal{I}$};
\draw[thin]
(0,2)--(4,6);
\node at (-0.55,2) {$\tau=T$};
\draw[thin]
(4,4) sin (3,5.5);
\node at (2.8,5.5) {$\mathcal{H}$};
\draw[dashed]
(-0.2,1.3)--(1.6,1.3);
\node at (2.1,1.3) {$S_{t_\star}(\uppercase\expandafter{\romannumeral1})$};
\draw[thin]
(0,1.2)--(1.2,1.2);
\draw[->]
(1.2,0.8)--(0.6,1.1);
\node at (1.65,0.8) {$S_{t_\star-\epsilon}$};
\draw[thin]
(0,1.2) sin (0.3,1.5);
\draw[thin]
(1.2,1.2) sin (0.9,1.5);
\draw[->]
(-0.3,0.8)--(0.5,1.4);
\node at (-0.8,0.7) {\tiny{$\mathcal{D}^+(S_{t_\star-\epsilon})$}};
\draw[dashed]
(0.8,3)--(3.3,3);
\node at (3.85,3) {$S_{t_\star}(\uppercase\expandafter{\romannumeral2})$};
\draw[thin]
(0.9,2.9)--(2.9,2.9);
\draw[dashed]
(2.3,4.5)--(3.85,4.5);
\node at (1.6,4.5) {$S_{t_\star}(\uppercase\expandafter{\romannumeral3})$};
\end{tikzpicture}
\end{center}
\caption{}
\label{fig:1}
\end{figure}

We have three case of $t_\star$ (see figure \ref{fig:1}).  For case $I$. First, by the standard theory of linear wave equations, we have local existence for \eqref{wave-equ} in $\overline{\mathcal{D}^+(S_{t_\star-\epsilon})}$ the domain of dependence of $S_{t_\star-\epsilon}$. And then, by the local existence of null-timelike problem, we have local existence in the domain bounded by $\mathcal{T}$ and $\partial(\mathcal{D}^+(S_{t_\star-\epsilon}))$. Last, by the local existence of characteristic initial value problem, we have local existence in the region bounded by $\partial(\mathcal{D}^+(S_{t_\star-\epsilon}))$ and $\Sigma_0$. Hence, for $\epsilon$ sufficiently small, there exists $\epsilon'>\epsilon$, such that the solution exists in
$\mathcal{M}_1\cap\Omega_T\cap\cup_{l\leq t_\star-\epsilon+\epsilon'}S_l$, which yield the global existence. For case $II$, we only need the standard theory and the local existence of characteristic  initial value problem. For case $III$, by the standard theory, we can get the same result.
\end{proof}

\begin{remark}
We have the existence of solution and $H^k$-estimates in $\mathcal{M}_1\cap\Omega_T$, and by Sobolev embedding, we can extend solution to $\mathcal{H}\cap\Omega_T$. Then by the same method, we can prove the existence of solution in $\mathcal{M}_2\cap\Omega_T$, which can be considered as a null-timelike problem with initial data setting on $\mathcal{H}$ and $\mathcal{I}$.
\end{remark}
\begin{theorem}
For $\varphi\in H^{2k}(\Sigma_0)$, $\psi_1\in H^{2k}(\mathcal{T})$, and $\psi_2\in H^k(\mathcal{I})$, then there exists a unique  solution of \eqref{wave-equ}-\eqref{wave-boundary} $v\in H^k(\Omega_T)$.
\end{theorem}

\section{Discussion}
In this paper, we consider the null-time-like boundary value problem of linear wave equations in general asymptotic AdS space-time. We show that the solution is globally existent and unique if the data is give on the time-like, null and conformal boundary. Further more, we also found similar result also exist for Maxwell field in an asymptotic AdS space-time and the paper is under preparing. These two cases are toy models from the holographic condense matter theory. In fact, the standard holographic model for condense matter theory contains scalar field, Maxwell field and gravity field\cite{3H08}, so we need to consider the coupled system of scalar field, Maxwell field and linearized gravity. This is a very important and interesting mathematical problem and will be considered in future work.

\section*{Acknowledgement}
This work is supported by the Natural Science Foundation of China
(NSFC) under Grant Nos. 11575286 and 11731001.


\begin{thebibliography}{ss}

\bibitem{CIV}H. Friedrich, Proc. Roy. Soc. Lond. {\bf A} {\bf 378}
(1981) 401; Proc. Roy. Soc. Lond. {\bf A} {\bf
381} 361;\\ J. Luk, \emph{On the local existence for the characteristic initial value problem in general
relativity}, arXiv:1107.0898 [gr-qc].

\bibitem{IBVP}O. Sarbach and M. Tiglio,
\emph{Continuum and Discrete Initial-Boundary Value Problems and Einstein¡¯s Field Equations},
Living Rev. Relativity, 15, (2012), 9. [Online Article]:
http://www.livingreviews.org/lrr-2012-9

\bibitem{Bartnik97}R. Bartnik,  \emph{Einstein equations in the null quasi-spherical gauge}, Class. Quant. Grav. {\bf 14} (1997) 2185.

\bibitem{Ma98}J. M. Maldacena, \emph{The Large N Limit of Superconformal Field Theories and Supergravity}, Adv. Theor. Math. Phys.{\bf 2} (1998) 231.

\bibitem{Wit98}E. Witten, \emph{Anti De Sitter Space And Holography}, Adv. Theor. Math. Phys. {\bf 2} (1998) 253.

\bibitem{Wit18}E. Witten, \emph{A Note On Boundary Conditions In Euclidean Gravity}, arXiv : 1805.11559.

\bibitem{ADSQCD}M. Kruczenski, D. Mateos, R. C. Myers and D. J. Winters,
JHEP 0307 (2003) 049 ; J. Babington, J. Erdmenger, N. J. Evans, Z. Guralnik and I. Kirsch, Phys. Rev. {\bf D 69} (2004) 066007.

\bibitem{3H08} S. A. Hartnoll, C. P. Herzog and G. T. Horowitz, \emph{Building an AdS/CFT superconductor}, Phys. Rev. Lett. {\bf 101} (2008) 031601.

\bibitem{ADSCMT}G. T. Horowitz, \emph{Introduction to Holographic Superconductors}, Lect. Notes Phys. {\bf 828} (2011) 313;\\
R.-G. Cai, L. Li, L. F. Li and R.-Q. Yang, \emph{Introduction to Holographic Superconductor models}, Sci. China Phys. Mech. Astron. {\bf 58} (2015) 060401.

\bibitem{Strominger}I. Bredberg, C. Keeler, V. Lysov and A. Strominger, \emph{From Navier-Stokes To Einstein}, JHEP 07 (2012) 146.

\bibitem{Alinhac} S. Alinhac,  \emph{Geometric analysis of hyperbolic differential equations: an introduction},
London Mathematical Society Lecture Note Series, 374. Cambridge University Press, Cambridge, 2010.

\bibitem{Alinac2004} S. Alinhac,
\emph{Remarks on energy inequalities for wave and Maxwell equations on a curved background},
Math. Ann., 329(2004),  707-722.

\bibitem{Balean1997} R. Balean, \emph{The null-timelike boundary problem for the linear wave equation},
Comm.  P. D. E.,  22(1997), 1325-1360.

\bibitem{Balean1998} R. Balean, R. Bartnik,  \emph{The null-timelike boundary problem for Maxwell's equations in Minkowski space.} R. Soc. Lond. Proc. Ser. A Math. Phys. Eng. Sci.  454(1998),  no. 1976, 2041-2057.

\bibitem{Baskin3} D. Baskin, A. S\'{a} Barreto,
\emph{A support theorem for a nonlinear radiation field},
Microlocal methods in mathematical physics and global analysis, 111-112, Trends Math.,
Birkh\"{o}user/Springer, Basel, 2013.

\bibitem{Baskin2} D. Baskin, A. S\'{a} Barreto,  \emph{Radiation fields for semilinear wave equations}
Trans. Amer. Math. Soc., 367 (2015),  3873-3900.

\bibitem{Baskin1} D. Baskin, A. Vasy, J. Wunsch,   \emph{Asymptotics of radiation fields in asymptotically Minkowski space}, Amer. J. Math., 137 (2015), 1293-1364.

\bibitem{Wang}  D. Baskin, F. Wang,  \emph{Radiation fields on Schwarzschild spacetime},
Comm. Math. Phys., 331(2014), 477-506.

\bibitem{Bondi1962} H., Bondi,  M. G. J. van der Burg,  A. W. K. Metzner,
\emph{Gravitational waves in general relativity. VII. Waves from axi-symmetric isolated systems},
Proc. Roy. Soc. Ser. A,  269(1962), 21-52.

\bibitem{Christodoulou1991} D. Christodoulou, S. Klainerman,  \emph{Asymptotic properties of linear field equations in Minkowski space.} Comm. Pure Appl. Math.  43(1990),  no. 2, 137-199.

\bibitem{Christodoulou1993} D. Christodoulou, S. Klainerman, \emph{The global nonlinear stability of the Minkowski space.} Princeton Mathematical Series, 41. Princeton University Press, Princeton, NJ, 1993.
\bibitem{Chandrasekhar1992} S. Chandrasekhar,  \emph{The mathematical theory of black holes.} Reprint of the 1992 edition. Oxford Classic Texts in the Physical Sciences. The Clarendon Press, Oxford University Press, New York, 1998.

\bibitem{Chrusciel1995} P. T. Chru\'{s}ciel, M. A. H. MacCallum, D. B. Singleton, \emph{Gravitational waves in general relativity. XIV. Bondi expansions and the ``polyhomogeneity'' of $\mathscr{I}$.} Philos. Trans. Roy. Soc. London Ser. A  350  (1995),  no. 1692, 113-141.

\bibitem{Chrusciel2002} P. T. Chru\'{s}ciel, E. Delay, \emph{Existence of non-trivial, vacuum, asymptotically simple spacetimes.} Classical Quantum Gravity  19  (2002),  no. 9, L71-L79.
\bibitem{Duff1958} G. F. D. Duff,
\emph{Mixed problems for linear systems of first order equations},
Canad. J. Math.,  10(1958), 127-160.

\bibitem{Friedlander1962} F. G. Friedlander,
\emph{On the radiation field of pulse solutions of the wave equation},
Proc. R. Soc. Lond. A, 269(1962),  53-65.

\bibitem{Friedlander1964} F. G. Friedlander,
\emph{On the radiation field of pulse solutions of the wave equation. II},
Proc. R. Soc. Lond. A, 279(1964),  386-394.

\bibitem{Friedlander1967} F. G. Friedlander,
\emph{On the radiation field of pulse solutions of the wave equation. III},
Proc. R. Soc. Lond. A, 299(1967),  264-278.

\bibitem{Friedlander1973} F. G. Friedlander,  \emph{An inverse problem for radiation fields.} Proc. London Math. Soc. (3)  27(1973), 551-576.

\bibitem{Friedlander1980} F. G. Friedlander,
\emph{Radiation fields and hyperbolic scattering theory},
Math. Proc. Camb. Phil. Soc., 88(1980), 483-515.

\bibitem{Friedlander2001} F. G. Friedlander,
\emph{Notes on the wave equation on asymptotically Euclidean manifolds},
J. Funct. Ana., 184(2001), 1-18.
\bibitem{Friedrich1998} H. Friedrich, \emph{Gravitational fields near space-like and null infinity.} J. Geom. Phys.  24  (1998),  no. 2, 83-163.
\bibitem{Friedrich2003} H. Friedrich, \emph{Spin-2 fields on Minkowski space near spacelike and null infinity.} Classical Quantum Gravity  20  (2003),  no. 1, 101-117.

\bibitem{Kroon2016} J. A. Valiebte Kroon, \emph{Conformal methods in general relativity.} Cambridge Monographs on Mathematical Physics. Cambridge University Press, Cambridge, 2016.
\bibitem{Kroon2017} E. Gasper\'{i}n, J. A. Valiente Kroon,
\emph{Polyhomogeneous expansions from time symmetric initial data.} Classical Quantum Gravity  34  (2017),  no. 19, 195007, 30 pp.
\bibitem{ZX2} H. Ge, M. Luo, Q. Su, D. Wang,  X. Zhang,
\emph{Bondi-Sachs metrics and photon rockets},
Gen. Relativity Gravitation, 43(2011),  2729-2742.

\bibitem{Hagen1977} Z. Hagen, H. J. Seifert,
\emph{On characteristic initial-value and mixed problems},
Gen. Relativity Gravitation, 8(1977),  259-301.

\bibitem{ZX} W. Huang, S.-T. Yau, X. Zhang,
\emph{Positivity of the Bondi mass in Bondi's radiating spacetimes},
Atti Accad. Naz. Lincei Cl. Sci. Fis. Mat. Natur., 17(2006), 335-349.

\bibitem{LinZhang2017} Q. Han, L. Zhang, \emph{Asymptotics for null-timelike boundary problems for general linear wave equations},  arXiv:1801.02795.

\bibitem{Klainerman2013} S. Klainerman, F. Nicol\`{o}, \emph{Peeling properties of asymptotically flat solutions to the Einstein vacuum equations.} Classical Quantum Gravity  20  (2003),  no. 14, 3215-3257.
\bibitem{Klainerman2003-1} S. Klainerman, F. Nicol\`{o}, \emph{The evolution problem in general relativity.} Progress in Mathematical Physics, 25. Birkh\"{a}user Boston, Inc., Boston, MA, 2003.
\bibitem{Melrose} R. Melrose, A. S\'{a} Barreto,  A. Vasy,
\emph{Asymptotics of solutions of the wave equation on de Sitter-Schwarzschild space},
Comm. P. D. E., 39(2014),  512-529.
\bibitem{Mason2009} L. J. Mason, J. P. Nicolas, \emph{Regularity at space-like and null infinity.} J. Inst. Math. Jussieu  8  (2009),  no. 1, 179-208.
\bibitem{Mason2012} L. J. Mason, J. P. Nicolas, \emph{Peeling of Dirac and Maxwell fields on a Schwarzschild background.} J. Geom. Phys.  62  (2012),  no. 4, 867-889.
\bibitem{Penrose1962} E. Newman, R. Penrose, \emph{ An approach to gravitational radiation by a method of spin coefficients.} J. Mathematical Phys.  3(1962) 566-578.

\bibitem{Penrose1965} R. Penrose,  \emph{Zero rest-mass fields including gravitation: Asymptotic behaviour.} Proc. Roy. Soc. Ser. A  284(1965) 159-203.
\bibitem{Rendall1990} A.D. Rendall, \emph{Reduction of the characteristic initial value problem to the Cauchy problem and its applications to the Einstein equations.} Proc. Roy. Soc. London Ser. A 427 (1990), no. 1872, 221-239.

\bibitem{A6}  A. S\'{a} Barreto, \emph{Radiation fields on asymptotically Euclidean manifolds},
Comm. P. D. E., 28(2003), 1661-1673.

\bibitem{A3} A. S\'{a} Barreto,
\emph{Radiation fields, scattering, and inverse scattering on asymptotically hyperbolic manifolds},
Duke Math. J., 129(2005),  407-480.

\bibitem{A2} A. S\'{a} Barreto,
\emph{A support theorem for the radiation fields on asymptotically Euclidean manifolds},
Math. Res. Lett., 15(2008),  973-991.


\bibitem{A} A. S\'{a} Barreto,
\emph{A local support theorem for the radiation fields on asymptotically euclidean manifolds},
J. Anal. Math., 130(2016), 275-286.

\bibitem{A4} A. S\'{a} Barreto,  J. Wunsch,
\emph{The radiation field is a Fourier integral operator},
Ann. Inst. Fourier (Grenoble), 55(2005),  213-227.

\bibitem{Sachs1961}  R. K. Sachs,  \emph{Gravitational waves in general relativity. VI. The outgoing radiation condition.} Proc. Roy. Soc. Ser. A  264(1961) 309-338.

\bibitem{Sachs1962} R. K. Sachs,
\emph{Gravitational waves in general relativity. VIII. Waves in asymptotically flat space-time},
Proc. Roy. Soc. Ser. A,  270(1962), 103-126.

\bibitem{Wang2} F. Wang,
\emph{Radiation field for Einstein vacuum equations with spacial dimension n$\ \geq\ $4}, arXiv:1304.0407.



\bibitem{ZX3} F. Xie, X. Zhang,
\emph{Peeling property of Bondi-Sachs metrics for nonzero cosmological constant},  arXiv:1704.06015.

\end{thebibliography}
\end{document}